\documentclass[letterpaper, 10 pt, conference]{ieeeconf}  % Comment this line out if you need a4paper

\IEEEoverridecommandlockouts                              % This command is only needed if 
\usepackage{graphics} % for pdf, bitmapped graphics files
\usepackage{epsfig} % for postscript graphics files
\usepackage{mathptmx} % assumes new font selection scheme installed
\usepackage{amsmath} % assumes amsmath package installed
\usepackage{amssymb}  % assumes amsmath package installed
\usepackage{algorithm}
\usepackage{algorithmic}

\usepackage{cite}

\usepackage[english]{babel}
\usepackage{dsfont}
\usepackage{hyperref}
\hypersetup{
    colorlinks=true,
    linkcolor=blue,
    filecolor=magenta,  
    urlcolor=cyan
    }
    
\newtheorem{theo}{Theorem}

\newtheorem{lemma}{Lemma}

\newtheorem{defn}{Definition}

\DeclareMathOperator{\trace}{tr}
\DeclareMathOperator{\vect}{vec}

\newcommand{\norm}[1]{\|#1\|}

\title{\LARGE \bf
Initial Excitation-based Adaptive Observers for Discrete-Time LTI Systems%{$^{\#}$}
}

\author{Anchita Dey*, Soutrik Bandyopadhyay and Shubhendu Bhasin% <-this % stops a space
\thanks{Anchita Dey, Soutrik Bandyopadhyay and Shubhendu Bhasin  are with the
  Department of Electrical Engineering, Indian Institute of Technology Delhi,
  Hauz Khas, New Delhi, Delhi 110016, India {\tt\small
    anchita.dey@ee.iitd.ac.in, soutrik.bandyopadhyay@ee.iitd.ac.in, sbhasin@ee.iitd.ac.in.}}%%
\thanks{*Corresponding author.}
}

\begin{document}

\maketitle
\thispagestyle{empty}
\pagestyle{empty}

%%%%%%%%%%%%%%%%%%%%%%%%%%%%%%%%%%%%%%%%%%%%%%%%%%%%%%%%%%%%%%%%%%%%%%%%%%%%%%%%
\begin{abstract}
In practical applications, the efficacy of a control algorithm relies critically on the accurate knowledge of the parameters and states of the underlying system. However, obtaining these quantities in practice is often challenging. Adaptive observers address this issue by performing simultaneous state and parameter estimation using only input–output measurements. While many adaptive observer designs exist for continuous-time systems, their discrete-time counterparts remain relatively unexplored. This paper proposes an initial excitation (IE)–based adaptive observer for discrete-time linear time-invariant systems. In contrast to conventional designs that rely on the persistence of excitation condition, which requires continuous excitation and infinite control effort, the proposed method does not require excitation for infinite time, thus making it more practical for stabilization tasks. We employ a two-layer filtering structure and a normalized gradient descent-based update law for learning the unknown parameters.  We also propose modifying the regressors to enhance information extraction, leading to faster convergence. Rigorous theoretical analysis guarantees bounded and exponentially converging estimates of both states and parameters under the IE condition, and simulation results validate the efficacy of the proposed design.

% In many cases, designing suitable controllers hinges on the availability of
% system parameter knowledge and state measurements. Unfortunately, it is
% challenging to obtain either of these accurately in practical scenarios. There
% have been several efforts to design controllers considering the absence of
% parametric knowledge and/or state information. Adaptive observers are a useful
% tool to solve this challenging problem since they carry out simultaneous state
% and parameter estimation using only the input and output data. While there
% exist many adaptive observer designs for continuous-time systems, the design
% for discrete-time systems is limitedly explored. In this work, we propose an
% adaptive observer for discrete-time linear time-invariant systems using an
% initially exciting input. The proposed method uses two layers of filters and a
% normalized gradient descent-based update law for learning the system state and
% parameters; in addition, matrix regressors that contain more information are
% used to get improved convergence. Unlike a persistently exciting signal that
% requires sufficient energy for all time, the initially exciting signal carries
% limited energy for a finite time until a condition on the regressors is met.
% Further, this condition is online verifiable and does not depend on future
% values of the signal. The proposed novel initial excitation-based adaptive
% observer guarantees convergence of the state and parameter estimates to their
% true values.

\end{abstract}

%%%%%%%%%%%%%%%%%%%%%%%%%%%%%%%%%%%%%%%%%%%%%%%%%%%%%%%%%%%%%%%%%%%%%%%%%%%%%%%%

\section{Introduction}

Accurate estimation of unmeasurable system states and parameters is critical for effective controller operation. Several applications of control theory, such as robotics, chemical processes, and industrial automation, rely on understanding the internal dynamics of the system that are often inaccessible to direct sensor measurements. In such cases, reliable state and parameter estimation is of paramount importance. 

The problem of state estimation from measurable output observations has been studied extensively for decades. Classical observer designs, such as the Luenberger observer \cite{luenberger2003observers}, have proven highly effective in systems with known dynamics.  While classical observer designs have shown reasonable success in the literature, in practice, it is difficult to design such observers when system parameters are unknown.
%TODO: The observers designed in [], [], can handle systems with unknwon parameters
% However, when params are unknown/uncertain, it is to difficult to design such 
In light of this, adaptive observers
\cite{luders1973adaptive,carroll1973adaptive, anderson1974adaptive,kreisselmeier1977adaptive,marine2001robust,zhang2002adaptive,tomei2022enhanced} were
introduced, which simultaneously estimate the system state and parameters online.

%TODO: Ei jaygatay add some literature on different adaptive observers, kontay ki korchhe.

%Such as maximum designs were for SISO, and in discrete time the designs are even less explored. MIMO discrete e amader ager ekta kaaj achhe, jetay naki PE chharao bounded thake signals, kintu convergence hobe ki na guarantee nei. convergence er jonnyo PE chai. eikhane IE diye extend kora hoche. Then explain briefly what happens here. And also we introduce Matrix regressors for discrete-time (literature pabi continuous time er). These help in faster convergence since they retain mre information.

%% MAtrix Regressor er upor pore nis olpo egulo theke : 
%G. Kreisselmeier, “Adaptive observers with exponential rate of convergence,” IEEE transactions on automatic control, vol. 22, no. 1, pp. 2–8, 1977
%L. G. Kraft, “Controllable form state variables obtained from an arbitrarily fast adaptive observer,” in 1976 IEEE Conference on Decisionand Control including the 15th Symposium on Adaptive Processes.IEEE, 1976, pp. 1088–1094.
%B. Jenkins, A. M. Annaswamy, and A. Kojic, “Matrix regressor adaptive observers for battery management systems,” in IEEE International
%Symposium on Intelligent Control (ISIC), 2015, pp. 707–714.
% D. Soudbakhsh and A. M. Annaswamy, “Health monitoring with matrix regressor based adaptive observers,” in IEEE 56th Annual Conference on Decision and Control (CDC), 2017, pp. 5131–5136

A crucial requirement for desirable convergence performance of the adaptive observers is that the input-output signals of the system satisfy the persistence of excitation (PE) condition \cite{sastry2011adaptive}. The PE condition ensures that the measured input-output data contains sufficient information to uniquely identify the uncertain parameters. However, the PE condition is difficult to verify online since it requires the excitation to persist indefinitely, necessitating an infinite control effort and dependence on future signal values.

To obviate this challenge, recent works have explored alternative approaches that relax this restrictive assumption. One promising approach is the Initial Excitation (IE)-based method \cite{roy2017combined} which utilizes finite yet informative excitation during the initial phases of the learning process to ensure parameter convergence {{exponentially.  }}This approach has found reasonable success in the literature in solving adaptive {{control and estimation}} problems \cite{roy2019robustness,jha2019initial, katiyar2022initial}. Unlike the PE condition, which requires the input signal to remain sufficiently rich for all time
\cite{boyd1986necessary}, the IE condition only demands finite excitation over a limited time interval. Once the IE condition is satisfied, adequate information about the system dynamics has been captured, and the excitation can be removed thereafter. Consequently, under IE, the system input can decay to zero while still ensuring convergence of the parameter and state estimation errors. In contrast, PE-based approaches necessitate continuous, albeit possibly small, perturbations to maintain excitation, thereby preventing the system states from settling exactly at the equilibrium.

{{While IE-based adaptive control and identification schemes have been explored in continuous time, their extension to discrete-time systems remains limited \cite{adejc}. The same statement applies to adaptive observer designs, which are mostly explored for continuous-time systems \cite{luders1973adaptive, carroll1973adaptive, kreisselmeier1977adaptive,zhang2002adaptive,anderson1974adaptive,marine2001robust,tomei2022enhanced} in contrast to the limited literature for discrete-time systems, such as \cite{kudva1974discrete,tamaki1981design,suzuki1980discrete, Pan02122023, DEY20238708}. The observers designed in \cite{kudva1974discrete,tamaki1981design,suzuki1980discrete} are for single-input single-output systems, whereas \cite{DEY20238708, Pan02122023} are for multi-input multi-output (MIMO) systems. The method proposed in \cite{Pan02122023} requires knowledge of the sets containing the state and parameters to be estimated. In \cite{DEY20238708}, we developed an adaptive observer for discrete-time MIMO linear time-invariant (LTI) systems that guarantees bounded parameter and state estimates under stabilizing inputs; however, convergence to the true parameters required a PE signal.

In this work, we extend the framework of \cite{DEY20238708} to design an IE-based adaptive observer for discrete-time LTI systems. Two layers of filters are used to express the input-output data in a suitable linear regression form. Unlike \cite{kreisselmeier1977adaptive,suzuki1980discrete,DEY20238708}, where the relation between the output measurements and the unknown parameter has a bounded perturbation-like term owing to the unknown initial state, in this design, we estimate the initial state as well. This allows us to express the output in terms of a modified parameter (that is made of the original parameter and the initial state) without any perturbation-like term in the linear regression. Furthermore, motivated by matrix-regressor formulations in continuous-time adaptive observers \cite{kreisselmeier1977adaptive,kraft1976controllable,jenkins2015matrix,soudbakhsh2017health}, we propose a similar concept of modifying the regressors in the discrete-time setting, such that the regressors retain richer information from input–output data, resulting in faster convergence and improved estimation performance. Accordingly, we introduce the concept of strong-IE for matrix regressors in the discrete-time setting, which is milder than the IE condition for vector regressors \cite{adejc}. A normalized gradient descent-based adaptation law with a switching component is used to estimate the unknown parameter and initial state, which are then used to obtain an estimate of the system state. The switching mechanism is enabled once the strong-IE condition is satisfied. The design ensures exponential convergence of the estimates to their true values without requiring PE. }}

\textit{Contributions:} The key contributions of this work are as follows. We propose a novel IE–based adaptive observer for discrete-time LTI systems that enables online simultaneous state and parameter estimation without requiring PE, thereby enhancing practical applicability. A two-layer filtering structure with extended regressors is developed to extract richer information from input–output data. {{The normalized gradient descent–based update law guarantees exponential convergence of both state and parameter estimation errors under the IE condition.}} Finally, rigorous theoretical analysis and simulation studies validate the effectiveness and convergence properties of the proposed observer design.

%%%%%%%%%%%%%%%%%%%%%%%%%%%%%%%%%%%%%%%%%%%%%%%%%%%%%%%%%%%%%
% \newpage
\textit{Notations:}
%Add the IE definition here. Refer to ABhishek da's paper on discrete-time IE. Label the definition as \label{Def1}.
$I_n\in \mathbb{R}^{n \times n}$ denotes the identity matrix. The zero matrix is denoted by $0_n\in\mathbb R^{n\times n}$ and by $0_{n\times m}\in\mathbb R^{n\times m}$. We represent the set of integers from $q$ to $\infty$ by $\mathbb{I}_q^+$ and $\mathbb{I}_q^r \triangleq \{q, q+1, \dots, r-1,r \}$ for integers $q$ and $r$, where $q < r$. For a matrix $A  \triangleq\begin{bmatrix}
  A_1 & A_2 & \dots & A_r
\end{bmatrix}\in \mathbb{R}^{q \times r}
$, where $A_i \in \mathbb{R}^q$ $ \forall i \in \mathbb{I}_1^r$, we define the operators $\vect(A) \triangleq \begin{bmatrix}
  A_1^\intercal & A_2^\intercal & \dots & A_r^\intercal \end{bmatrix}^\intercal \in \mathbb{R}^{qr}
$, and tr$(A)$ to be the trace of the matrix $A$. The symbols $\lambda_{\min} (\cdot) $ and $\lambda_{\max} (\cdot) $ denote the minimum and maximum eigenvalues of a matrix, respectively. We use $\otimes$ to denote the Kronecker product, $\| \cdot \|$ to denote the Euclidean norm of a vector, and $\|\cdot \|_{\text{F}}$ for representing the Frobenius norm of a vector or matrix.

\section{Preliminaries}%
We recall the PE condition for the discrete-time setup from the adaptive control literature \cite{adejc}.
\begin{defn}\label{defn:pe} 
A bounded signal $\phi_t \in \mathbb{R}^{h}$ is said to be \emph{PE} if there exists a fixed integer $t_{PE} \in \mathbb{I}_1^+$, where $t_{PE} \geq h$, and a constant scalar
$\zeta > 0$ such that \end{defn}
\begin{align*}
  \sum_{i=0}^{t_{PE}-1} \phi_{t+i} \phi_{t+i}^\intercal \geq \zeta I_h\;\;\;\; \forall t \in \mathbb{I}_0^+.
\end{align*}
While the PE condition provides a necessary and sufficient condition for exponential convergence of parameter estimates in the adaptive control literature, it is worth noting that this condition is dependent on the future behavior of the dynamical system, making it non-trivial to verify the PE condition online.

To obviate this challenge, the condition of IE was introduced in \cite{roy2017combined}. We now recall the definition of the discrete-time IE condition \cite{adejc}. 
\begin{defn}\label{defn:nie}
A bounded signal $\phi_t \in \mathbb{R}^{h}$ is said to be \emph{IE} if there exists a finite integer $t_{IE} \in \mathbb{I}_1^+$, where $t_{IE} \geq h$, and a constant scalar $\zeta > 0$ such that in the time interval $[0,t_{IE}]$ \end{defn}
\begin{align}\label{eqIE}
  \sum_{i=0}^{t_{IE}-1} \phi_i \phi_i^\intercal \geq \zeta I_h.
\end{align}

Motivated by \cite{fastaoatul} to use regressors that retain more information for faster convergence, we propose the following definition of the strong-IE condition.
{{\begin{defn}\label{defn:ie}
A bounded signal $\phi_t \in \mathbb{R}^{h \times m}$, where $m>1$, is said to be \emph{strongly initially exciting} (SIE) if there exists a finite integer $t_{SIE} \in \mathbb{I}_1^+$, where $t_{SIE} \geq h$, and a constant scalar $\zeta > 0$ such that in the time interval $[0,t_{SIE}]$ 
\end{defn}
\begin{align}\label{eqSIE}
  \sum_{i=0}^{t_{SIE}-1} \phi_i \phi_i^\intercal \geq \zeta I_h.
\end{align}}}

%TODO: It would be satisfied only once right?

Unlike the PE condition, the IE and SIE conditions require only the current and past measurements of the signal $\phi_t$, thereby enabling online verification.

% For discrete time, the concept of Uniform strong initial excitation may not exist. The regressor is not explicitly dependent on time. See the Atul sir paper. They use the concept of uniformity to talk about non-autonomous regressor. 

% I think we can safely skip that.

\section{Problem Formulation}

Consider a discrete-time LTI system 
\begin{align}
    &x_{t+1}=Ax_t+Bu_t, \label{sys1}\\
    &y_t=Cx_t, \label{ytrue}
\end{align}
that is completely observable and has the following canonical realization
\begin{align}
    A=\begin{bmatrix}
       A^{(1)} & I_q & 0_q&... &0_q \\
       A^{(2)} & 0_q & I_q&... &0_q\\
       \vdots & \vdots & \vdots & \ddots & \vdots\\
       A^{(r-1)} & 0_q & 0_q&... &I_q\\
       A^{(r)} & 0_q & 0_q&... &0_q
    \end{bmatrix}, \; C=\begin{bmatrix}
        I_q & 0_q&... &0_q
    \end{bmatrix},\label{AC_canon}
\end{align}
with input $u_t\in\mathbb R^m$, output $y_t\in\mathbb R^q$, state $x_t\in\mathbb R^{n}$ (where $n=qr$), and unknown parameters $A^{(i)}\in\mathbb R^{q\times q}$ $\forall i\in\mathbb{I}_1^r$ and $B\in\mathbb R^{n\times m}$. In addition to the unknown parameters $A,\;B$, the state measurements are also unavailable. The objective is to simultaneously learn the unknown terms $A,\;B$ and $x_t$ using the available input-output data. 

While \cite{DEY20238708} also considers the same objective, the obtained parameter estimates are guaranteed to be bounded with a stabilizing and bounded input, and the parameter estimation error converges to a finite value; the parameter and state estimation errors converge to zero only in the presence of a sufficiently rich input \cite{boyd1986necessary}. However, a sufficiently rich input or PE regressor requires constant energy for all time, making this a restrictive choice. On the other hand, regressors that are IE require a certain amount of energy for a finite time until a condition (\eqref{eqIE} or \eqref{eqSIE}) is met, after which the excitation may be removed; the condition involves the collection of sufficient information in the regressors that is needed to perform successful identification, and further, the satisfaction of the condition can be determined online. These factors motivated us to design a suitable IE-based adaptive observer for estimating the states and parameters. {{The regressors, introduced subsequently, are matrices by design; therefore, the SIE condition \eqref{eqSIE} is used. Accordingly, we use `IE' and `SIE' interchangeably in the following sections.}}

%%%%%%%%%%%%%%%%%%%%%%%%%%%%%%%%%%%%%%%
\section{The Adaptive Observer Design}
The plant dynamics in \eqref{sys1} can be rewritten as
\begin{align}
    x_{t+1}=Fx_t+(A-F)x_t+Bu_t,\label{sys2}
\end{align}
where $F\in\mathbb R^{n\times n}$ is a user-defined Schur stable matrix structurally similar to $A$ in \eqref{AC_canon}, i.e., 
\begin{align}
    F=\begin{bmatrix}
       F^{(1)} & I_q & 0_q&... &0_q \\
       F^{(2)} & 0_q & I_q&... &0_q\\
       \vdots & \vdots & \vdots & \ddots & \vdots\\
       F^{(r-1)} & 0_q & 0_q&... &I_q\\
       F^{(r)} & 0_q & 0_q&... &0_q
    \end{bmatrix}.
\end{align}
Let $a\triangleq \vect\left( \begin{bmatrix}
    {A^{(1)}}^\intercal & {A^{(2)}}^\intercal & ... & {A^{(r)}}^\intercal
\end{bmatrix}^\intercal \right)\in\mathbb R^{qn}$, and  $b\triangleq \vect(B)\in\mathbb R^{mn}$, where $a$ and $b$ are vectors containing the unknown terms, and $f\triangleq \vect \left( \begin{bmatrix}
    {F^{(1)}}^\intercal & {F^{(2)}}^\intercal & ... & {F^{(r)}}^\intercal
\end{bmatrix}^\intercal \right)\in\mathbb R^{qn}$. These are combined to define the unknown parameter vector $p\triangleq \begin{bmatrix}(a-f)^\intercal & b^\intercal \end{bmatrix}^\intercal\in\mathbb R^{qn+mn}$ and its estimate $\hat p_t$. The plant dynamics in \eqref{sys2} can be further modified using the parameter vectors as follows.
\begin{align}
    x_{t+1}&=Fx_t+(a-f)y_t+Bu_t=Fx_t+Z_tp,
\end{align}
where $Z_t\triangleq \begin{bmatrix}        I_n\otimes y_t^\intercal  & I_n\otimes u_t^\intercal       \end{bmatrix}\in\mathbb R^{n\times (qn+mn)}$.
We are interested in estimating the parameter $p$ and the state $x_t$ using the input and output measurements. To this end, we define filter matrices in the following subsections that enable us to form a linear regression using the input-output data and the unknown parameter vector; the regression is then leveraged to design a suitable update law.

\subsection{The First Layer of Filter}
Consider the following dynamics for a filter variable $M_t\in\mathbb R^{n\times (qn+mn)}$
\begin{align}\label{FilterM}
    M_{t+1}=FM_t+Z_t\;,\;\;M_0=0_{n\times (qn+mn)}.
\end{align}
This allows us to write the plant state and output as
\begin{align}
    &x_t=M_tp+F^tx_0, \text{ and }  \label{xreg}\\
    &y_t=\phi_t^\intercal p+CF^t x_0=\underbrace{\begin{bmatrix}
        \phi_t^\intercal & CF^t
    \end{bmatrix}}_{=:w_t} \underbrace{\begin{bmatrix}
        p^\intercal & x_0^\intercal
    \end{bmatrix}^\intercal}_{=:\psi}   \label{yreg}
\end{align}
respectively, where the regressors $\phi_t\triangleq M_t^\intercal C^\intercal\in\mathbb R^{(qn+mn)\times q}$, $w_t\in\mathbb R^{q\times (qn+mn+n)}$ {{and the unknown parameter $\psi\in\mathbb R^{qn+mn+n}$. Although the term $\psi$ contains the initial state $x_0$, we treat it as a constant parameter and attempt to learn it along with the plant parameter $p$. This allows writing \eqref{yreg} as a linear regression without any perturbation-like term $CF^tx_0$, which otherwise existed in the regression expression when $x_0$ was not being estimated (see \cite{DEY20238708, kreisselmeier1977adaptive}). }}

The expression $y_t=w_t\psi$ is in the linear regression form, with $\psi$ being unknown. Let the estimate of $\psi$ be denoted by $\hat \psi_t$. Designing a suitable adaptation law for $\hat \psi_t$ yields the parameter estimates $\hat p_t$ as well as estimates of the initial state $ \hat x_{{0_t}}$ using
\begin{align}\label{extractpx}
\begin{aligned}
   & \hat p_t= \begin{bmatrix}
        I_{qn+mn} & 0_{(qn+mn)\times n}\\0_{n\times (qn+mn)} & 0_n
    \end{bmatrix}\hat \psi_t,\\
    &\hat x_{0_t}= \begin{bmatrix}
        0_{qn+mn} & 0_{(qn+mn)\times n}\\0_{n\times (qn+mn)} & I_n
    \end{bmatrix}\hat \psi_t.
    \end{aligned}
\end{align} 
The adaptive observer state and output are defined using \eqref{xreg} and \eqref{yreg} as
\begin{align}
  &  \hat x_t=M_t\hat p_t+F^t \hat x_{{0_t}}=\begin{bmatrix}
        M_t & F^t
    \end{bmatrix}\hat \psi_t ,\;\;\text{and }  \label{xhat}\\
 &   \hat y_t=w_t\hat \psi_t,  \label{yhat}  
\end{align}
respectively. 

While an adaptive law for $\hat \psi_t$ can be developed using \eqref{yreg} and \eqref{yhat}, we propose the use of larger regressor matrices that contain more information to allow faster convergence of the estimates to their true values. 

\subsection{Modification of the Linear Regression}
In this work, we propose the use of a normalized gradient descent-based adaptive law. The law is developed with an aim to minimize the output estimation error $\tilde y_t\triangleq y_t-\hat y_t$. {{However, the convergence speed of the observer can be arbitrarily increased by including an additional $qn+mn+n-1$ independent error signals \cite{kreisselmeier1977adaptive,lion1967rapid}. To this end, we define the following matrices.     \begin{align}\label{defWYt}\begin{aligned}
   & W_t\triangleq \begin{bmatrix}
    w_{t_{(0)}}^\intercal  &  w_{t_{(1)}}^\intercal  &   ... &  w_{t_{(qn+mn+n-1)}}^\intercal 
\end{bmatrix}^\intercal  \\&\;\;\;\;\;\;\;\;\;\;\;\;\;\;\;\;\;\;\;\;\;\;\;\;\;\;\;\;\;\;\;\;\;\;\;\;\;\;\;\;\;\;\;\;\;\;\;\in\mathbb R^{q(qn+mn+n)\times(qn+mn+n)},  \\
& Y_t\triangleq \begin{bmatrix}
    y_{t_{(0)}}^\intercal  &  y_{t_{(1)}}^\intercal  &   ... &  y_{t_{(qn+mn+n-1)}}^\intercal 
\end{bmatrix}^\intercal  \in\mathbb R^{q(qn+mn+n)}, 
\end{aligned}
\end{align}
where the elements are generated as
\begin{align}\label{MatReg}\begin{aligned}
&w_{t_{(0)}}\triangleq w_t,\; \;w_{0_{(i)}}\triangleq 0_{q\times(qn+mn+n)}\;\;\forall i\in\mathbb I_1^{qn+mn+n-1},\\
& w_{t_{(i)}}\triangleq \alpha w_{{t-1}_{(i)}}+(1-\alpha)w_{{t-1}_{(i-1)}} \;\;\forall (t,i)\in\mathbb {I}_1^+\times\mathbb I_1^{qn+mn+n-1}, \\
&y_{t_{(0)}}\triangleq y_t,\; \;y_{0_{(i)}}\triangleq 0_{q}\;\;\forall i\in\mathbb I_1^{qn+mn+n-1},\\
& y_{t_{(i)}}\triangleq \alpha y_{{t-1}_{(i)}}+(1-\alpha)y_{{t-1}_{(i-1)}} \;\;\forall (t,i)\in\mathbb {I}_1^+\times\mathbb I_1^{qn+mn+n-1},  
\end{aligned}
\end{align}
where $\alpha\in (0,1)$. Comparing \eqref{yreg} with \eqref{defWYt} and \eqref{MatReg} allows us to write
\begin{align}
    Y_t=W_t\psi.\label{YNewreg}
\end{align}}}
 {{The $qn+mn+n$ variables in $Y_t$ lead to $qn+mn+n$ error signals in $Y_t-W_t\hat\psi_t$. Given sufficient excitation, the stacking of the past regressors with the current one creates possibilities of spanning the $qn+mn+n$-dimensional subspace faster. The update law for $\hat \psi_t$ is subsequently developed using this new regression \eqref{YNewreg} along with another filter variable discussed next.}}

\subsection{The Second Layer of Filter}
The incorporation of IE is done as a switched adaptive mechanism using another layer of filter variables defined as
\begin{align}
  &  S_{t+1}=\sigma S_t+W_t^\intercal W_t,\;\;\;S_0=0_{qn+mn+n}, \label{Sreg}\\
   & \rho_{t+1}=\sigma \rho_t+W_t^\intercal Y_t,\;\;\;\rho_0=0_{(qn+mn+n)\times 1},\label{rhoreg}
\end{align}
 where $\sigma \in (-1,1)$. The variables $S_t$ and $\rho_t$ satisfy 
 \begin{align}
     \rho_t=S_t \psi,  \label{rho_S}
 \end{align}
 which is another linear regression involving the unknown vector $\psi$. 
 
 \subsection{SIE Condition}
 A suitable adaptation law for learning $\hat \psi_t$ is developed leveraging three linear regression components - two of which use \eqref{YNewreg} and \eqref{rho_S}. The third component in the update law is also based on the regression \eqref{rho_S}; however, it corresponds to the satisfaction of the SIE condition defined next. The complete update law is given in Sec.~\ref{ParaLawSec}. 
 
 Based on Definition \ref{defn:ie}, the bounded signal $W_t$ is SIE if $\exists$ a fixed integer $t_{SIE}\geq qn+mn+n$, and a constant $\zeta>0$ such that
\begin{align}\label{IEcond}
    \sum_{i=0}^{t_{SIE}-1} W_i^\intercal W_i\geq \zeta I_{qn+mn+n}.
\end{align}

\subsection{Parameter Estimation Law}\label{ParaLawSec}
The parameter $\hat \psi_t$ is estimated using the following normalized gradient descent-based update law.
\begin{align}
    \hat \psi_{t+1}=&\hat\psi_t+ \{\kappa_1W^\intercal_{t+1}(Y_{t+1}-W_{t+1}\hat\psi_t) + \kappa_2 S_{t+1}^\intercal (\rho_{t+1}-\nonumber\\
        &S_{t+1}\hat\psi_t)+
    \kappa_3 \eta S_{t_{SIE}}^\intercal (\rho_{t_{SIE}}-S_{t_{SIE}}\hat\psi_t)\}/\{1+\trace(\Gamma_t^\intercal \Gamma_t)\} \nonumber\\&
    \;\;\;\;\;\;\;\;\;\;\;\;\;\;\;\;\;\;\;\;\;\;\;\;\;\;\;\;\;\;\;\;\;\;\;\;\;\;\;\;\;\;\;\;\;\;\;\;\;\;\;\;\;\;\;\;\;\;\;\;\;\;\;\;\;\;\;\;\;\;\forall t\in\mathbb I_0^+\label{eq:update_law},\\
    \text{where }&\Gamma_t\triangleq \begin{bmatrix}
        \sqrt \kappa_1 W_{t+1}^\intercal & \sqrt \kappa_2 S_{t+1}^\intercal & \sqrt \kappa_3 \eta S_{t_{SIE}}^\intercal
    \end{bmatrix}^\intercal, \label{eq:Gamma_definition}\\
    & \eta=\begin{cases}
        &1,\;\;\forall t\geq t_{SIE}\\
        &0,\;\;\text{otherwise}.
    \end{cases}\label{etalaw}
\end{align}

Once initialized with estimates of state and parameters, the adaptive observer takes in the input and output data and provides online simultaneous estimates of the parameter{{ $\hat \psi_t$ using \eqref{eq:update_law}-\eqref{etalaw} from which we can then obtain the estimates $\hat p_t$ and $\hat x_{0_t}$ and $\hat x_t$ using \eqref{extractpx} and \eqref{xhat}.}} Algorithm \ref{alg} gives the steps for implementing the proposed observer.
\begin{algorithm}[h!]
\caption{IE-based Adaptive Observer}\label{alg}
\begin{algorithmic}[1]
\REQUIRE $n$, $q$, $m$, $\hat x_{0_0}$, $\hat p_0$, $F$, $\alpha$, $\sigma$, $\kappa_1$, $\kappa_2$, $\kappa_3$, $u_t$.
\ENSURE  $\hat p_t$, $\hat x_{0_t}$, $\hat x_t$ $\forall t \in\mathbb I_1^+$.\\
% {{\textbf{Offline Steps:} }}
\STATE Measure $y_0$. Compute $f$, and define $M_0$, $S_0$, $\rho_0$, $\hat \psi_0$. Set $\eta=0$, $S_{t_{SIE}}=0$ and $\rho_{t_{SIE}}=0$. 
\STATE Apply input $u_0$. 
\FOR{$t\geq 1$} 
\STATE Measure output $y_t$ of the plant given by \eqref{sys1}-\eqref{ytrue}.
\STATE Compute $M_t$, $w_t$, $W_t$, $Y_t$, $S_t$, $\rho_t$, $\Gamma_t$ using \eqref{FilterM}, \eqref{yreg}, \eqref{defWYt}-\eqref{rhoreg} and \eqref{eq:Gamma_definition}. Also, compute $\sum_{i=0}^{t-1}W_{i}^\intercal W_i$.
\IF{ {{($ \eta=0$) $\land$ (\eqref{IEcond} holds true)}}}
\STATE Set $t_{SIE}=t$, $S_{t_{SIE}}=S_t$, $\rho_{t_{SIE}}=\rho_t$ and $\eta=1$.
\ENDIF
\STATE Update $\hat \psi_t$ using \eqref{eq:update_law}.
\STATE Compute $\hat p_t$, $\hat x_{0_t}$, $\hat x_t$ using \eqref{extractpx} and \eqref{xhat}.
\STATE Apply $u_t$ to the plant \eqref{sys1}.
\STATE Update $t\leftarrow t+1$.
\ENDFOR
\end{algorithmic}
\end{algorithm}

%%%%%%%%%%%%%%%%%%%%%%%%%%%%%%%

\section{Convergence Analysis}
We now state a consequence of the IE condition that would be useful for the proof of convergence of the parameter estimates.
\begin{lemma}
\label{lemma:pd} The regressor $W_t$ defined in \eqref{MatReg} is SIE if and only if the regressor $S_t$ defined in \eqref{Sreg} is positive definite for some $t=t_{SIE}$, where $t_{SIE}\geq qn+mn+n$, and $n,\;q,\;m$ are the number of states, outputs and inputs of the plant mentioned in \eqref{sys1}-\eqref{ytrue}.
\end{lemma}
\begin{proof}
{{The proof follows from \cite[Lemma 1]{adejc}.}}
\end{proof}

We now state the main result of the paper.
\begin{theo}
 {{ Consider the plant dynamics given by \eqref{sys1}-\eqref{ytrue}, and let the input $u_t$ be a bounded and stabilizing signal such that the resulting regressor $W_t$ in \eqref{MatReg} is bounded for all $t\geq 0$ and satisfies the SIE condition in \eqref{IEcond}. Then, the switched parameter update law in \eqref{eq:update_law}-\eqref{etalaw} ensures that the parameter estimation error $\tilde{\psi}_t$ and consequently, the state estimation error $\tilde x_t\triangleq x_t-\hat x_t$, exponentially converges to its respective origin.}}
  % Suppose the input $u_t$ to the plant in \eqref{sys1}-\eqref{ytrue} is bounded and stabilizing, resulting in the regressor $W_t$ in \eqref{MatReg} to be bounded. If the regressor $W_t$  is IE, then the switched parameter update law in
  % \eqref{eq:update_law} ensures that the parameter and state estimation errors $\tilde{\psi}_t$ exponentially converges to the origin.
\end{theo}
\begin{proof}
  We can re-write the update law in \eqref{eq:update_law} as
  \begin{align}
    \hat{\psi}_{t+1} =
    \hat{\psi}_t + \frac{\Gamma_t^{\intercal} \Gamma_t \tilde{\psi}_t}{1+\trace(\Gamma_t^\intercal \Gamma_t)}
    =
    \hat{\psi}_t + \gamma_t^\intercal \gamma_t \tilde{\psi}_t,
    \label{eq:update_law_unmeasurable}
  \end{align}
where $\gamma_t \triangleq \frac{\Gamma_t}{\sqrt{1+\trace(\Gamma_t^\intercal \Gamma_t)}}$
and $\Gamma_t$ is defined in \eqref{eq:Gamma_definition}. Subtracting both sides of
\eqref{eq:update_law_unmeasurable} from $\psi$, we obtain the dynamics of the
parameter estimation error as
\begin{align}
  \tilde{\psi}_{t+1} = [I_{qn + mn + n} - \gamma_t^\intercal \gamma_t]\tilde{\psi}_t.
  \label{eq:psi_tilde}
\end{align}
We now consider the candidate Lyapunov function $V_{t} \triangleq \| \tilde{\psi}_{t} \|_{\text{F}}^2 = \trace(\tilde{\psi}_{t}^{\intercal}\tilde{\psi}_{t})$. Using the dynamics of the parameter estimation error from \eqref{eq:psi_tilde} and leveraging the properties of the trace operator, we can write
  \begin{align*}
    V_{t+1} = \trace(\tilde{\psi}_{t+1}^{\intercal}\tilde{\psi}_{t+1})
  = \trace(\tilde{\psi}_{t}^{\intercal}\tilde{\psi}_{t})
  + \trace(-2 \epsilon_t^\intercal \epsilon_t + \epsilon_t \epsilon_t^\intercal \gamma_t \gamma_t^\intercal),
  \end{align*}
where $\epsilon_t \triangleq \gamma_t \tilde{\psi}_t$. Subtracting $V_t$ from both sides yields
  \begin{align*}
    V_{t+1} - V_t =
- \trace(\epsilon_t^\intercal \epsilon_t )  + \trace(-\epsilon_t^\intercal \epsilon_t + \epsilon_t \epsilon_t^\intercal \gamma_t \gamma_t^\intercal).
  \end{align*}
  Utilizing the fact that $\trace(\epsilon_t^\intercal \epsilon_t) = \norm{\epsilon_t}_{\text{F}}^2$, we have
  \begin{align*}
    \Delta V_{t+1} \leq
  - \trace(\epsilon_t^\intercal \epsilon_t )  - \norm{\epsilon_t}_{\text{F}}^2+ \norm{\epsilon_t}_{\text{F}}^2 \norm{\gamma_t}_{\text{F}}^2,
  \end{align*}
  where $\Delta V_{t+1} \triangleq V_{t+1} - V_t$. Note that
  $\norm{\gamma_t}_{\text{F}}^2 = \frac{\norm{\Gamma_t}_{\text{F}}^2}{1+ \norm{\Gamma_t}_{\text{F}}^2} < 1$. We thus obtain the bound
  \begin{align}
    \Delta V_{t+1} \leq
  - \trace(\epsilon_t^\intercal \epsilon_t ) = -\trace(\tilde{\psi}_t^\intercal \gamma_t^\intercal \gamma_t \tilde{\psi}_t). \label{eq:last_step}
  \end{align}
If the regressor $W_t$ is SIE, then $S_{t_{SIE}}$ is positive-definite (Lemma \ref{lemma:pd}) and $\eta = 1 \;\; \forall \; t\geq t_{SIE}$. Thus, for all $t\geq t_{SIE}$, the   matrix $\gamma_t^\intercal \gamma_t$ is positive-definite. In other words, there exists a positive constant $\kappa$ such that $\kappa \leq \lambda_{\min}(\gamma_t^\intercal \gamma_t) \;\;\forall t \geq t_{SIE}$. This implies that
  \begin{align}
  & \Delta V_{t+1} \leq -\trace(\tilde{\psi}_t^\intercal \gamma_t^\intercal \gamma_t \tilde{\psi}_t) \leq 
  - \lambda_{\min}(\gamma_t^\intercal \gamma_t) \norm{\tilde{\psi}_t}_{\text{F}}^2 \nonumber\\
   \Rightarrow & \;V_{t+1}-V_t\leq -\kappa V_t \Rightarrow  V_{t+1}\leq (1-\kappa) V_t. \label{VtKappa}
  \end{align}
Further, since $\|\gamma_t\|_{\text{F}}^2 < 1$ and  $\gamma_t^\intercal \gamma_t$ is symmetric positive definite under the SIE assumption, we can write
\begin{align}
   0 < \lambda_{\min}(\gamma_t^\intercal \gamma_t) \leq \lambda_{\max}(\gamma_t^\intercal \gamma_t) \leq 1   
    &\Rightarrow \kappa\in(0,1] \nonumber\\
    &\Rightarrow (1-\kappa)\in[0,1). \label{kapparange}
\end{align}
Therefore, from \eqref{VtKappa} and \eqref{kapparange}, it is proved that the parameter estimation error $\tilde{\psi}_t$ exponentially converges to the origin. 

The state estimation error is given by $\tilde x_t=\begin{bmatrix}
    M_t & F^t
\end{bmatrix}\tilde \psi_t$ from \eqref{xreg} and \eqref{xhat}. Since $F$ is Schur-stable, and $M_t$ evolves following \eqref{FilterM}, we can conclude that with a stabilizing and bounded input $u_t$ and the estimation error $\tilde\psi_t$ exponentially converging to origin, the state estimation error $\tilde x_t$ also exponentially converges to the respective origin.  
\end{proof}

%%%%%%%%%%%%%%%%%%%%%%%%%%%
\section{Simulation Result}
Consider the following LTI system:
\begin{align*}
    &x_{t+1}=\begin{bmatrix}0.4 & 1 & 0\\ 0.5 & 0 & 1\\-0.1 & 0 & 0\end{bmatrix}x_t+\begin{bmatrix}0.1 & -0.2\\0.2 & 0.1 \\0.3 & 0\end{bmatrix}u_t,\\
    &y_t=\begin{bmatrix}
        1 & 0 & 0
    \end{bmatrix}x_t.
\end{align*}
The initial terms are arbitrarily chosen as $\hat A_0=\begin{bmatrix}
    5 & 1 & 0;\;5 & 0 & 1;\;5 & 0 & 0 
\end{bmatrix}$, $\hat B_0=\begin{bmatrix}
   5 & 5;\;5 & 5;\;5& 5
\end{bmatrix}$, $x_0=\begin{bmatrix}
    1&1&1
\end{bmatrix}^\intercal$, and $\hat x_{0_0}=\begin{bmatrix}
    0.9&0.9&0.9
\end{bmatrix}^\intercal$. The observer gain and tuning parameters are selected as $F=\begin{bmatrix}
    0.0022 & 1 & 0;\;0.011 & 0 & 1;\;0.0001 & 0 & 0
\end{bmatrix}$, $\kappa_1=$ $\kappa_2=1.05$, $\kappa_3=0.01$, $\alpha=0.26$ and $\sigma=-0.98$. Figures \ref{fig:1} and \ref{fig:2} show that with the proposed method, the estimation errors $||\tilde\psi_t||_2$ and $||\tilde x_t||_2$ exponentially converge to zero. 

The results are compared with \cite{DEY20238708} using a PE and a non-PE input with the parameters $k_0=100$, $k_{\min}=0.1$ and $R=1$ for the recursive least square law in \cite{DEY20238708}. In Figs. \ref{fig:3} and \ref{fig:4}, we compare the 2-norm of the parameter estimation errors and of the state estimation errors, respectively, achieved using the proposed method and \cite{DEY20238708}; the applied inputs\footnote{The excitation component in the applied input is of the form $0.2u_{IE}[\sin(2t)+ u_{exc}\{\sin(3t)+\sin(5t)+\sin(7t)+\sin(11t)+\sin(13t)+\sin(17t)+\sin(23t)+\sin(29t)+\sin(31t)+\sin(37t)+\sin(41t)\};\;\sin(59t)+u_{exc}\{\sin(157t)+\sin(163t)+\sin(167t)+\sin(173t)+\sin(179t)+\sin(61t)+\sin(67t)+\sin(71t)+\sin(73t)+\sin(79t)+\sin(83t)\}]$, where (a) for the proposed method $u_{IE}=e^{0.001}$, $u_{exc}=1$, (b) for \cite{DEY20238708} with PE, $u_{IE}=1$, $u_{exc}=1$, and (c) for the non-PE case of \cite{DEY20238708}, $u_{IE}=1$, $u_{exc}=0$.} showing the excitation are plotted in Fig. \ref{fig:5}. We observe that the estimation errors converge to zero with the PE signal whereas the absence of sufficient excitation in the non-PE case leads to bounded non-zero estimation errors. On the other hand, the proposed method does not require consistent excitation (compare Fig. \ref{fig:5} (a) with (b) and (c)); the IE condition is satisfied at $t_{SIE}=12$, and it leads to exponential convergence of the estimation errors.

 \begin{figure}[ht!]
      \vspace{0.23cm}\centering
     \framebox{\parbox{3in}{\includegraphics[scale=0.418]{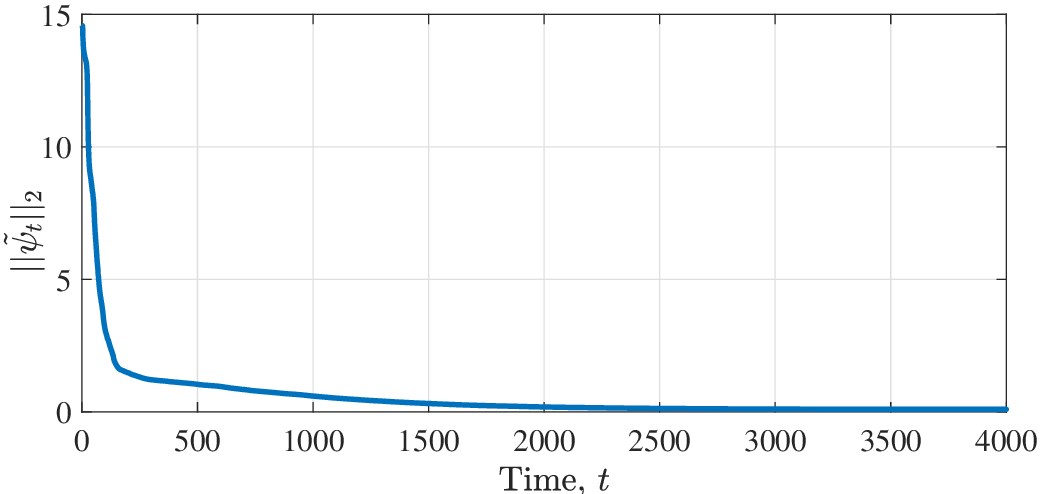}}}
      \caption{2-norm of the estimation error $\tilde \psi_t$ with the proposed method.}       \label{fig:1}
   \end{figure}
   \begin{figure}[h!]
      \vspace{0.23cm}\centering
     \framebox{\parbox{3in}{\includegraphics[scale=0.418]{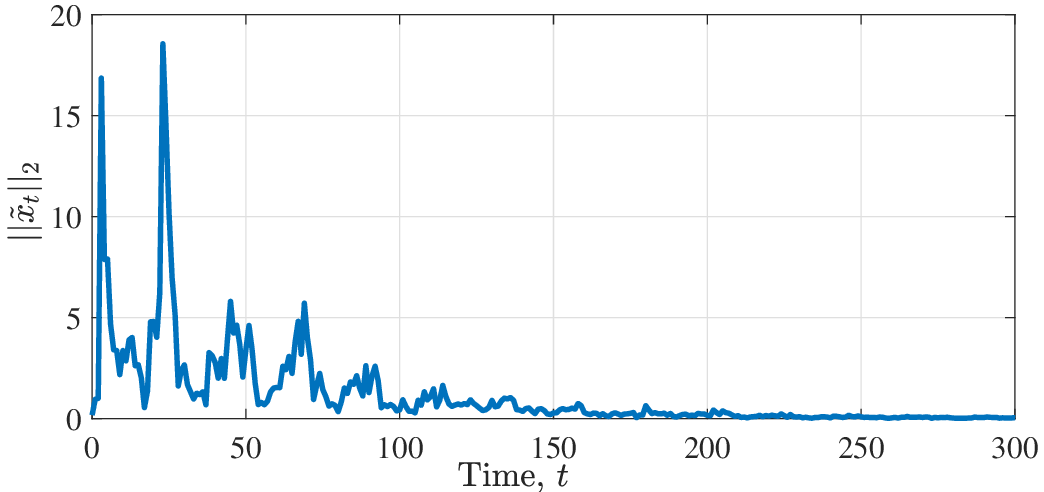}}}
      \caption{2-norm of the state estimation error with the proposed method.}       \label{fig:2}
   \end{figure}
   \begin{figure}[th!]
      \vspace{0.23cm}\centering
     \framebox{\parbox{3in}{\includegraphics[scale=0.418]{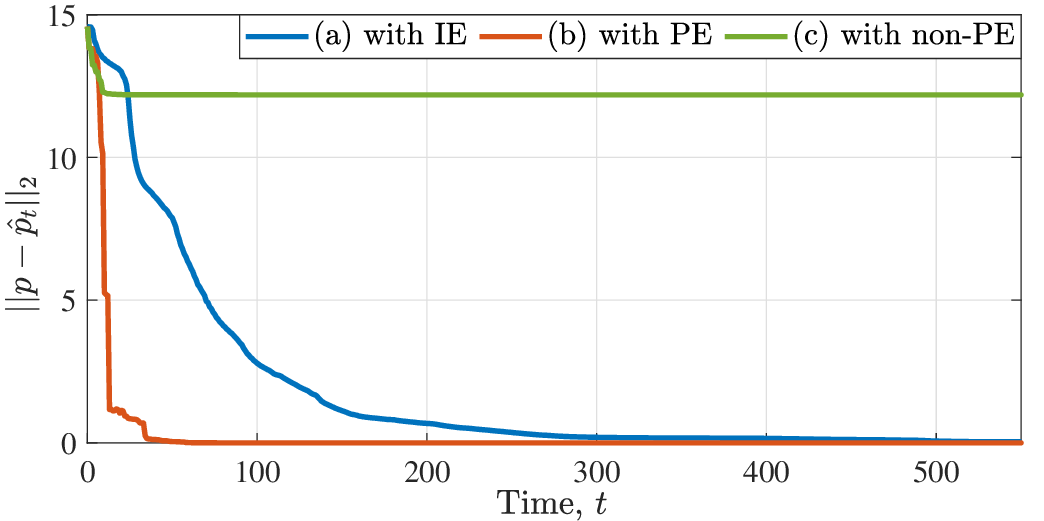}}}
      \caption{Comparison of the 2-norm of the parameter estimation errors using (a) the proposed method with an IE input, (b) \cite{DEY20238708} with a PE input, and (c) \cite{DEY20238708} with a non-PE input.}       \label{fig:3}
   \end{figure}
   \begin{figure}[th!]
      \vspace{0.23cm}\centering
     \framebox{\parbox{3in}{\includegraphics[scale=0.418]{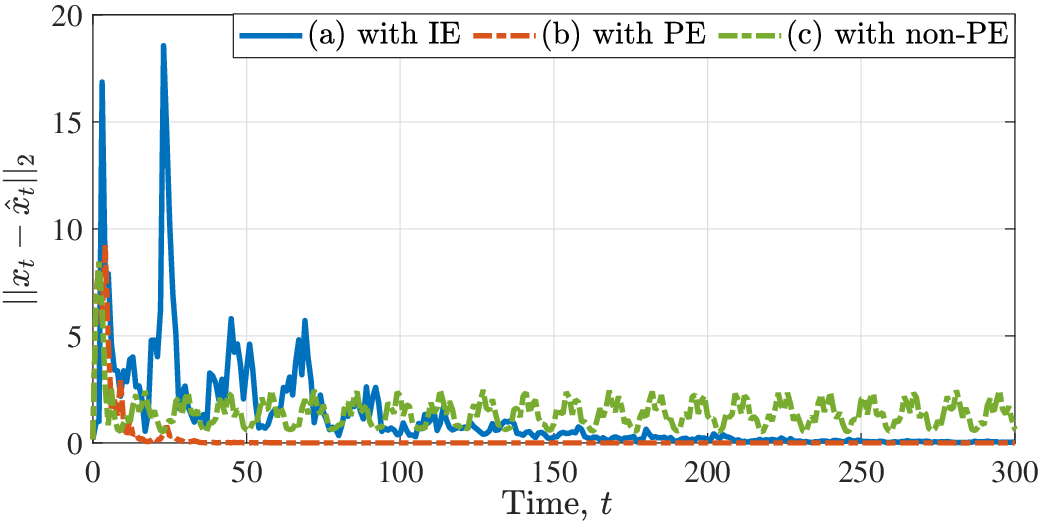}}}
      \caption{Comparison of the 2-norm of the state estimation errors using (a) the proposed method with an IE input, (b) \cite{DEY20238708} with a PE input, and (c) \cite{DEY20238708} with a non-PE input.}       \label{fig:4}
   \end{figure}
   \begin{figure}[th!]
      \vspace{0.23cm}\centering
     \framebox{\parbox{3in}{\includegraphics[scale=0.418]{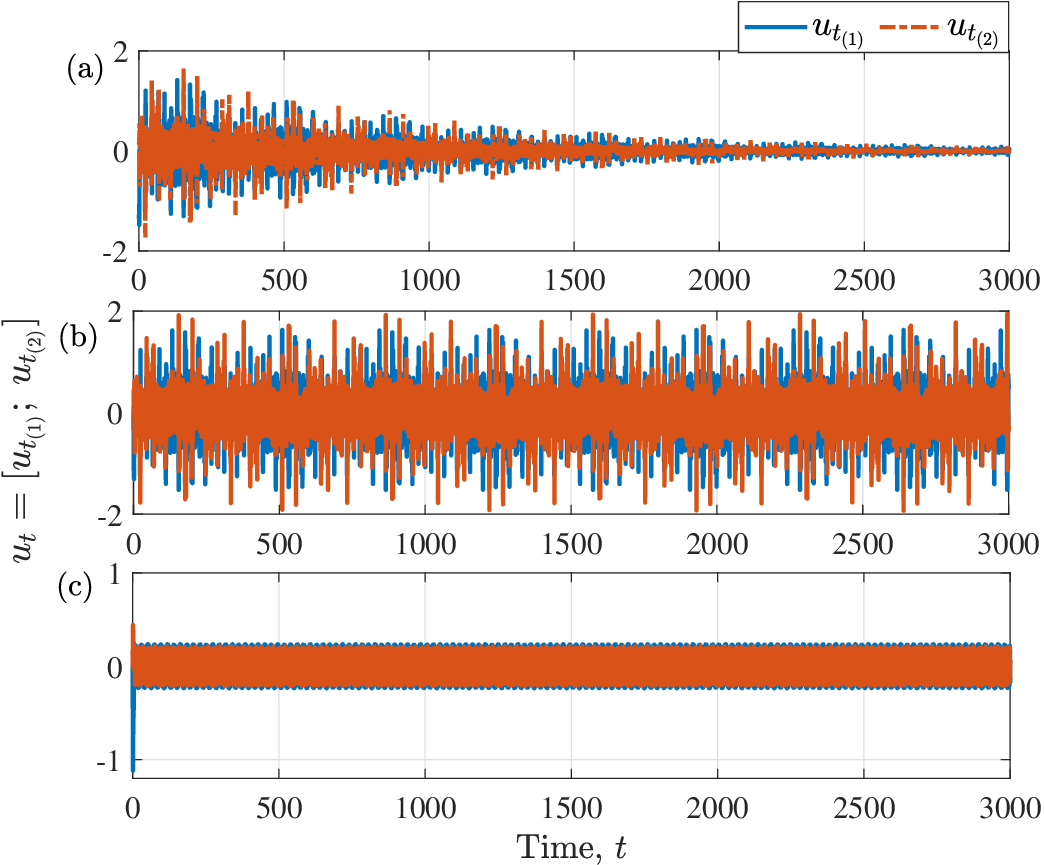}}}
      \caption{Input signal used for (a) the proposed method, (b) \cite{DEY20238708} with PE, and (c) \cite{DEY20238708} with non-PE.}       \label{fig:5}
   \end{figure}
%%%%%%%%%%%%%%%%%%%%%%%%%%%
% \addtolength{\textheight}{-12cm} 
\section{Conclusion}

In this work, we propose an IE-based adaptive observer for discrete-time LTI systems. The observer simultaneously estimates both the system states and unknown parameters using only input–output data, without requiring the restrictive PE condition. By incorporating a two-layer filtering structure and a normalized gradient descent–based update law, the design ensures stable and exponentially convergent estimation. The introduction of modified regressors with independent rows enables the extraction of richer information, resulting in faster convergence. Unlike traditional approaches that rely on the stringent PE condition, which demands continuous excitation and infinite control effort, the proposed method relies on a finite-time IE signal, making the convergence condition online verifiable and practically feasible. Theoretical analysis guarantees exponential convergence of the estimation errors under the IE condition, and simulation results validate the efficacy of the proposed observer in accurately estimating both states and parameters.

Future work will focus on extending the proposed framework to nonlinear and time-varying systems, investigating robustness to measurement noise, and exploring adaptive output-feedback controller designs based on the proposed observer.

\bibliographystyle{IEEEtran}
\footnotesize{\bibliography{IEEEabrv,reference}}
% \addtolength{}{-12cm}

\end{document}